\documentclass[american,a4paper,runningheads,envcountsect]{llncs}

\usepackage{babel}
\usepackage[utf8]{inputenc}
\usepackage[T1]{fontenc}

\usepackage{amsmath,amsthm,amssymb}
\usepackage{mathtools}
\usepackage{mathrsfs}
\usepackage{fca}
\usepackage{dsfont} 

\usepackage[colorlinks,citecolor=blue,urlcolor=black,hidelinks,linktocpage]{hyperref}
\usepackage[all]{hypcap}
\usepackage{cleveref}
\let\cref\Cref

\newtheorem{thm}{Theorem}[section]

\newtheorem{lem}[thm]{Lemma}

\theoremstyle{definition}

\theoremstyle{remark}

\renewcommand{\epsilon}{\varepsilon}
\renewcommand{\phi}{\varphi}

\newcommand{\ie}{i.\,e.\xspace}

\newcommand{\eg}{e.\,g.\xspace}
\usepackage[shrink=45,babel,kerning=true]{microtype}
\usepackage{csquotes}
\usepackage{booktabs}
\usepackage{paralist}
\usepackage[subtle]{savetrees}
\usepackage{todonotes}
\presetkeys{todonotes}{color=blue!5}{}

\usepackage[backend=bibtex,style=numeric-comp,doi=false,isbn=false,%
url=false]{biblatex}
\newcommand\blfootnote[1]{%
  \begingroup
  \renewcommand\thefootnote{}\footnote{#1}%
  \addtocounter{footnote}{-1}%
  \endgroup
}

\begin{document}

\title{Clones in Graphs}

\author{Stephan Doerfel\inst{1} \and Tom Hanika\inst{2, 3} \and Gerd Stumme\inst{2, 3}}

\date{\today}

\institute{%
  Micromata GmbH
  Kassel, Germany\\[0.5ex]
  \and
  Knowledge \& Data Engineering Group,
  University of Kassel, Germany\\[0.5ex]
  \and
  Interdisciplinary Research Center for Information System Design\\
  University of Kassel, Germany\\[0.5ex]
  \email{stephan.doerfel@doerfel.info, tom.hanika@cs.uni-kassel.de,
    stumme@cs.uni-kassel.de}
}
\maketitle

\blfootnote{Authors are given in alphabetical order.
  No priority in authorship is implied.}

\begin{abstract}
  Finding structural similarities in graph data, like social networks,
  is a far-ranging task in data mining and knowledge discovery. A
  (conceptually) simple reduction would be to compute the automorphism
  group of a graph. However, this approach is ineffective in data
  mining since real world data does not exhibit enough structural
  regularity. Here we step in with a novel approach based on mappings
  that preserve the maximal cliques. For this we exploit the well
  known correspondence between bipartite graphs and the data structure
  formal context $(G,M,I)$ from Formal Concept Analysis. From there we
  utilize the notion of clone items. The investigation of these is
  still an open problem to which we add new insights with this
  work. Furthermore, we produce a substantial experimental
  investigation of real world data. We conclude with demonstrating the
  generalization of clone items to permutations.
\end{abstract}


\section{Introduction}
\label{sec:introduction}


The identification of (structural) similar entities in graph data sets is a particularly relevant task in data analysis: it provides insights
into entities in the data (\eg, in members of social networks); it
allows grouping entities and even reducing data sets by removing redundant (structurally equivalent) elements (factorization).
For \emph{bipartite} graph data, a notion of structural similarity
that suggests itself is that of \emph{clone items}, known from the
realm of Formal Concept Analysis (FCA). The latter is a mathematical
toolset for qualitative data analysis, relying on algebraic notions
s.a.\ lattices and closure systems. Here, clone items are entities from the same partition that are completely interchangeable within the family of that partition's closed subsets.

In this paper, we follow up on a long-standing open problem of FCA,
collected at
ICFCA~2006,\footnote{\url{http://www.upriss.org.uk/fca/problems06.pdf}}
regarding the meaning of clone items in real world graph data.
The notion of clones was initially
proposed\footnote{This work is noted to be submitted (e.\,g, in~\cite{Gely05}), but has never been published.}
in ``Clone items: a pre-processing information for knowledge discovery'' by
R. Medina and L. Nourine. 
Subsequently, a plethora of desirable properties of clone items has been shown, such as,
\textquote{hidden combinatorics}~\cite{Gely05} that allow  
factorizations of data structures containing clones, 
computational properties
investigations, like~\cite{Medina05}, or the use of clones in 
%
association rule mining~\cite{Medina06}. 
Finally, the question of semantics was addressed by~\cite{Macko12}, who investigated clones in three well-known data sets
(\emph{Mushroom, Adults,} and \emph{Anonymous} from the UCI Machine Learning Repository~\cite{UCI}).
Following the observation that two data sets were free of clones whereas the
mushroom data set had only few, \cite{Macko12} introduced \emph{nearly clones} relying rather on statistical
than on structural properties.
However,
despite these previous efforts, the question -- are clone
items frequent in natural graph data sets -- in particular in social network data --
has not yet been answered in general.

The contributions of this paper are threefold:
First, we provide a prove for the
characterization of clone items on the level of formal contexts that allows us to easily compute clone items in
large data sets.
Second, we investigate a diverse variety of public realworld data sets
coming from different domains and exhibiting different properties. We
show that clones are not common in these data sets and conclude that in
their present form, clones are not as useful as one would have hoped,
regarding the efforts made in previous literature.
Third, to
%
%
%
resolve this dilemma, we point out a more general notion of clones.
For this we fall back to permutations on the
set of attributes in a formal context, providing a natural extension of the
clone property. These \emph{higher order clones} are able to identify more
complicated \textquote{clone structures} and should be the next step in
the investigation of relational data structures.

This work is structured as follows. In
Section~\ref{sec:form-conc-analys}, we recall basic notations of FCA
and show the correspondence to graphs. Then, in
Section~\ref{sec:theor-observ}, we provide a characterization of clone
items on the level of formal contexts. Following this, in
Section~\ref{sec:clon-soci-netw} we demonstrate how the notion of
clones can be applied in the realm of graphs. Subsequent to
experiments on various data sets, 
in Section~\ref{sec:higher-order-clones}, we extend the notion of
clone items to higher order clones. Eventually, we conclude our work
with Section~\ref{sec:conclusion}.

\section{Preliminaries}
\label{sec:form-conc-analys}
We give a short recollection of the ideas from formal concept analysis
as introduced in~\cite{Wille1982, fca-book} that are relevant in this
work. We use the common presentation of formal contexts by
$\context=\GMI$, where $G$ and $M$ are sets and
$I\subseteq G\times M$. The elements of $G$ are called objects, those
of $M$ are called attributes, and $(g,m)\in I$ signifies that object
$g$ has the attribute $m$. The correspondence to a bipartite graph
(network) is at hand. Let $H=(U\cup W,E)$ be such an undirected
bipartite graph with $U\cap W=\emptyset$ where $U$ is a set of
entities (often users), $W$ some set of common properties, and
$E\subseteq\{\{u,w\}\mid u\in U,w\in W\}$ the set of edges between $U$
and $W$.  There are two natural ways of identifying $H$ as a formal
context. In the following, we choose $\context(H)=(U,W,I)$ as the to
$H$ associated formal context,\footnote{The second way yields the dual
  context $\context(H)=(W,U,I)$.} where for $u\in U$ and $w\in W$, we
have $(u,w)\in I:\Leftrightarrow \exists e\in E:u\in e\wedge w\in
e$. For the case of a non-bipartite Graph $G=(V,E)$ we simply
construct the formal context by $\context=(V,V,I)$ with
$(u,v)\in I\Leftrightarrow \{u,v\}\in E$ for all $u,v\in V$. In the
following we use the terms network, (bipartite) graph, and formal
context interchangeably in the sense above.

We will utilize the common \emph{derivation} operators
$\cdot'\colon\mathcal{P}(G)\to\mathcal{P}(M), A\mapsto B\coloneqq\{m\in M\mid
\forall g\in A\colon (g,m)\in I\}$ and
$\cdot'\colon\mathcal{P}(M)\to\mathcal{P}(G), B\mapsto A\coloneqq\{g\in G\mid
\forall m\in B\colon (g,m)\in I\}$. Having those operations we call a formal
context $\context=\GMI$ \emph{object clarified} iff
$\forall g,h\in G, g\neq h: g'\neq h'$, \emph{attribute clarified} iff
$\forall m,n\in M,m\neq n:m'\neq n'$ and \emph{clarified} iff it is both. In
this definition we used $g'$ as shorthand for $\{g\}'$. Clarification will later
on correspond to a particular trivial kind of clones. Similarly we call a
clarified context $\context$ \emph{object reduced} if for all $g\in G$ there is
no $S\subseteq G\setminus\{g\}$ such that $g'=S'$. We call $\context$
\emph{attribute reduced} iff for all $m\in M$ there is no $S\subseteq
M\setminus\{m\}$ such that $m'=S'$. And, we call this $\context$ \emph{reduced}
iff $\context$ is attribute and object reduced.

A pair $(A,B)$ where $A\subseteq G$, $B\subseteq M$ with $A'=B$ and $B'=A$ is
called a \emph{formal concept}. Here, $A$ is called the \emph{concept extent}
and $B$ is called the \emph{concept intent}. The set of all these formal
concepts, i.e.,
$\mathfrak{B}(\context)\coloneqq \{(A,B)\mid A\subseteq G,B\subseteq
M,A'=B,B'=A\}$ gives rise to an order structure $(\mathfrak{B},\leq)$ using
$(A,B)\leq (C,D) :\Leftrightarrow A\subseteq C$, called \emph{concept
  lattice}. For clone items we are particularly interested in the two entailed
closure systems, i.e, in the \emph{object closure system}
$\mathfrak{G}(\context)\coloneqq \{A\in G\mid (A,B)\in\mathfrak{B}(\context)\}$
and the \emph{attribute closure system} $\mathfrak{M}(\context)\coloneqq \{B\in
M\mid (A,B)\in\mathfrak{B}(\context)\}$. We may denote those by $\mathfrak{G}$
and $\mathfrak{M}$ whenever the according context is implicitly given.

\subsubsection{Clones}
\label{sec:clones}
Besides the original definition of what clone items are there will be some
graduations useful to graphs. We start with the common definition.
Given a formal context $\context=\GMI$ and two items $a,b\in M$, we say $a$
\emph{is clone to} $b$ in $\mathfrak{M}$ if $\forall X\in\mathfrak{M}\colon
\phi_{a,b}(X)\in\mathfrak{M}$, with:
\[\phi_{a,b}(X)\coloneqq
  \begin{dcases}
    X\setminus\{a\}\cup\{b\}&\text{if}\ a\in X\wedge b\not\in X\\
    X\setminus\{b\}\cup\{a\}&\text{if}\ a\not\in X\wedge b\in X\\
    X&\text{else}
  \end{dcases}\]

We may denote this property by $a\sim_{\context} b$ and whenever the context is
distinctive $a\sim b$. It is obvious that $\sim$ is a reflexive and symmetric
relation on $M\times M$. Actually, it is also transitive, which can be shown
easily, hence $\sim$ is an equivalence relation. Since every $a\in M$ is a
clone to itself we say an $a$ is a \emph{proper clone} iff there is a $b\in
M\setminus\{a\}$ such that $a\sim b$. In a not-clarified formal context there
might be some $m,n\in M, m\neq n$ such that $m'=n'$. Those elements are proper
clones. However, this is obvious and not revealing any hidden structure besides
the fact that two identical copies are present. Therefore we call a proper clone
$a\in M$ \emph{trivial} iff there is a $b\in M\setminus\{a\}$ with $a'=b'$.

A this point one may ask if it is hard to construct a formal context having a
significant number of non-trivial clones. This is very easy as the following
example discloses.

\begin{example}
  The nominal scales, i.e., $(\{1,\dotsc,n\},\{1,\dotsc,n\},=)$ and
  the contra-nominal-scale $(\{1,\dotsc,n\},\{1,\dotsc,n\},\neq)$
  provide formal contexts where every attribute element is a
  non-trivial clone.  Furthermore, the union of two formal contexts,
  i.e., $\context_{1}\coloneqq(G_{1},M_{1},I_{1})$ and
  $\context\coloneqq(G_{2},M_{2},I_{2})$ becomes
  $\context_{1}\cup \context_{2}\coloneqq (G_{1}\cup G_{2},M_{1}\cup
  M_{2},I_{1}\cup I_{2})$, preserves the clones from $\context_{1}$
  and $\context_{2}$.
\end{example}


All the above can be defined similarly for elements of $G$ using the dual-context,
i.e., the context where objects and attributes are interchanged. We therefore
omit the explicit definitions and continue assuming the necessary definitions
are made. However, we may provide some wording to differentiate between clones
in $\mathfrak{M}$ and clones in $\mathfrak{G}$ for some formal context
$\GMI$. When necessary we call the former \emph{attribute clone} and the latter
\emph{object clone}.

\section{Theoretical observations}
\label{sec:theor-observ}

In this section, we derive some crucial properties of clones as well
as a characterization of the clone property on the level of the
context table. These theoretical results allow a fast computation of
clones and help understanding the nature of clones in data. The first
shows that for attributes with $a\sim b$ the object sets $a'$ and $b'$
are incomparable.

\begin{lem}[Clones are incomparable]\label{the:clones-incomparable}
  Let $\context=\GMI$ be a formal context and $a,b\in M$.
  If $a\sim b$, then from $a'\subseteq b'$ follows $a'= b'$.
\end{lem}
\begin{proof}
  Using $a'\subseteq b'$ we show $b'\subseteq a'$.
  We examine the mapping
  \begin{equation*}
    \phi_{ab}(b'')  =\begin{dcases}
      b'' & \text{if } a \in b'' \\
      b''\setminus\{b\}\cup \{a\} & \text{if}\ a\notin b''.\\
    \end{dcases}
  \end{equation*}
  We show that the second case is invalid. From $a\sim b$ and $\phi_{ab}(b'')$
  being a closure we deduce $a''\subseteq b''\setminus\{b\}\cup \{a\}$.  Since
  $a'\subseteq b'$, we have $b''\subseteq a''$ and together we yield
  $b\in b'' \subseteq a''\subseteq b''\setminus\{b\}\cup \{a\}$ contradicting
  the case. Hence, only the first case can exist, meaning $a \in b''$, thus
  obviously $b'\subseteq a'$.
\end{proof}

The next results indicates, that reducible elements of a formal context can be ignored in the search for clones.

\begin{lem}[Clone irreducability]
  \label{lem:clonirred}
  Let $\context=(G,M,I)$ be a clarified formal context and attributes $a,b\in M: a\neq b$ with  $a\sim b$.
  Then $a$ is irreducible in $\context$.
\end{lem}
\begin{proof}
  Assume $a$ is reducible, \ie, there exists a set of attributes $N\subseteq M$ with $a\notin N$ and $\bigcap_{n\in N} n' =  a'$.
  As $\context$ is clarified, we have $a'\neq b'$, thus from Lemma~\ref{the:clones-incomparable} follows $b\notin a''$. Therefore $\phi_{a,b}(a'') = a''\setminus\{a\}\cup \{b\}$. From the reducibility assumption follows
  $$\forall n\in N: n'\supseteq a' \Longrightarrow n\in a'' \stackrel{n\neq a}{\Longrightarrow} n\in a''\setminus\{a\}\cup \{b\} = \phi_{a,b}(a'').$$
  Thus, $a'=\bigcap_{n\in N}{n'} \supseteq \phi_{a,b}(a'')'$, which means $a''\subseteq\phi_{a,b}(a'')=a''\setminus\{a\}\cup \{b\}$. Clearly, this means $a=b$ contradicting the lemma's assumption. 
\end{proof}

While clarifying a context removes the non-trivial clones, additionally reducing that context does not change the clone relationship any further. Therefore, for finding non-trivial clones it suffices considering reduced contexts. Next, we describe for such contexts how clones can be identified directly from the context's table.
\cite{Gely05} already found it is sufficient to check join-irreducible intents to check the clone property. The respective result there (Proposition 1) is formulated for the dual version of formal contexts, \ie, where $G$ and $M$ are interchanged. Also, for the proof the authors of \cite{Gely05} refer to a manuscript that had been submitted (at the time) but appears to have never been published. For the sake of completeness, we present a variation of their result in the common notion of a formal context and present a proof. Here, we already use the fact that in a reduced context, the join irreducible concepts are exactly the object concepts.

\begin{theorem}\label{th:irreducible-intents}
  Let $\context=(G,M,I)$ be a reduced formal context and $a,b\in M$ with $a\neq b$. The following are equivalent:\begin{compactenum}
  \item   $a\sim b$ \label{th:irreducible-intents:clones}
    \item For each object $g\in G$, there is an object $h\in G$ such that $\phi_{a,b}(g') = h'$.\label{th:irreducible-intents:irreducible}
  \end{compactenum}
\end{theorem}

\begin{proof}
  First we show, $\ref{th:irreducible-intents:clones}.\Longrightarrow\ref{th:irreducible-intents:irreducible}.$
  For $a,b\in g'$ or $a,b\notin g'$, the claim is obvious (using $h\coloneqq g$). Without loss of generality, we can assume $a\in g'$ and $b\notin g'$, thus $\phi_{a,b}(g') = g'\setminus\{a\}\cup\{b\}$.

  As $\phi_{a,b}(g')$ is an intent, there exists a set of objects $H\subseteq G$ with $H'=\phi_{a,b}(g')=g'\setminus\{a\}\cup\{b\}$.
 We can partition $H$ into $H_{a}\coloneqq \{h\in H\mid a\in h'\}$ and $H_{\bar{a}}\coloneqq \{h\in H\mid a\notin h'\}$.
%
As clearly $a\notin \phi_{a,b}(g')$,  $H_{\bar{a}}$ cannot be empty.
  We yield:
  \begin{align*}
                           &                               & g'\setminus\{a\}\cup \{b\}                  & =  \bigcap_{h\in H_{a}}{h'} \cap \bigcap_{h\in H_{\bar{a}}}{h'} \\
     H_{\bar{a}}\neq\emptyset, b\notin g' &\Longrightarrow & g'\setminus\{a\}  & =  \bigcap_{h\in H_{a}}{h'} \cap \bigcap_{h\in H_{\bar{a}}}{(h'\setminus{\{b\}})}\\
     a\in g' &\Longrightarrow & g' & =  \bigcap_{h\in H_{a}}{h'} \cap \bigcap_{h\in H_{\bar{a}}}{(h'\setminus{\{b\}\cup \{a\}})}\\
     b\in\phi_{a,b}(g')=H' &\Longrightarrow & g' & =  \bigcap_{h\in H_{a}}{h'} \cap \bigcap_{h\in H_{\bar{a}}}{\phi_{a,b}(h')}
  \end{align*}
As $g$ is irreducible, we either have an object $h\in H_{a}$ with $g'=h'$ or an object $h\in H_{\bar{a}}$ with $g'=\phi_{a,b}(h')$.
Clearly, the former cannot be true, as $b\in h'$ for $h\in H$ and $b\notin g'$.
From the latter follows $\phi_{a,b}(g')=h'$.

Next, we show $\ref{th:irreducible-intents:irreducible}.\Longrightarrow\ref{th:irreducible-intents:clones}:$
Let $N\subseteq M$ be an intent of $\context$, \ie, there is a set of objects $H\subseteq G$ such that $N=H'$. We show that $\phi_{a,b}(N)$ is an intent. This is trivial for the cases $a,b\in N$ and $a,b\notin N$. Without loss of generality, we assume $a\in N$ and $b\notin N$.
Then \[\phi_{a,b}(N) = \phi_{a,b}(H') = H'\setminus\{a\}\cup\{b\} =\bigcap_{h\in
    H_{b}}(h'\setminus\{a\}\cup\{b\})\cap\bigcap_{h\in H_{\bar b}}(h'\setminus\{a\}\cup\{b\})\]
with $H_{b}$ and $H_{\bar{b}}$ defined as $H_{a}$ and
$H_{\bar a}$.
For $h\in H_{b}$ it holds $h'\setminus\{a\}\cup\{b\}=h'\setminus\{a\}$. 
As $H_{\bar b}\neq\emptyset$ (a.p., $b\not\in N=H'$) we yield
$\phi_{a,b}(N)=\bigcap_{h\in H_{b}}h'\cap\bigcap_{h\in H_{\bar
    b}}(h'\setminus\{a\}\cup\{b\})\,.$ Since $a\in H'$, for
$h\in H_{\bar{b}}$: $h'\setminus\{a\}\cup\{b\} = \phi_{a,b}(h')$,
which by~\ref{th:irreducible-intents:irreducible}. is $g'$ for some
$g\in G$. Thus $\phi_{a,b}(N)$ is the intersection of intents and
therefore itself an intent.

\end{proof}

The theorem characterizes clones on the context level: Two attributes $a$ and $b$ are clones if for each object $g\in G$ whose row contains only one of the two attributes, there is another object $h\in G$ such that its row contains only the other of the two attributes, while the remaining parts of the rows are identical, \ie, $g'\setminus\{a\} = h'\setminus\{b\}$.

%


\subsection{Clones in Graph Data}
\label{sec:experiments}
\subsubsection{Clones in Social Networks}
\label{sec:clon-soci-netw}
In the following, we identify any given graph with the formal context
counterpart $\context=(U,W,I)$, as
in~\cref{sec:form-conc-analys}. Transferring the definitions from
Section~\ref{sec:clones}, we obtain what clones in graphs, in
particular in social networks, are. For the special case of social
networks we call object clones \emph{user clones} and attribute clones
are either some \emph{property clone}, in the bipartite case, or also
user clones, in the single-mode case.

\begin{example}[Social Network]
  In \cref{fig:toyexample} we show a small artificial example of a
  possible social network. Represented as context as described
  in~\cref{sec:form-conc-analys} we get with
  $M=\{\text{\textbf{S}wimming, \textbf{H}iking, \textbf{B}iking,
    \textbf{R}afting, \textbf{J}ogging}\}$ the closure system
  $\mathfrak{M}(\context)=\{\{\text{S}\},\{\text{H}\},\{\text{B}\},
  \{\text{R}\},\{\text{J}\},\{\text{B,R}\}\},
  \{\text{B,J}\},\{\text{R,J}\}\}$. The associated clone classes are
  denoted in~\cref{fig:toyexample}.
\end{example}

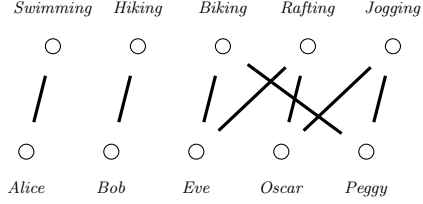
\begin{figure}[t]
  \centering
  \begin{tikzpicture}[scale=0.7,transform shape,node distance=.2cm and 3cm]
    \begin{scope}[individual/.style =
      {draw, circle, inner sep = 0.1cm}, label
      distance=0.3cm]
      \node[individual] (userA) [label=below:{\textit{Alice}},
      label=above:{\parbox{5cm}{}}]{};
      \node[individual] (userB) [right=1.3cm of userA,
      label=below:{\textit{Bob}}]{};
      \node[individual] (userC) [right=1.3cm of userB,
      label=below:{\textit{Eve}}]{};
      \node[individual] (userD) [right=1.3cm of userC,
      label=below:{\textit{Oscar}}]{};
      \node[individual] (userE) [right=1.3cm of userD,
      label=below:{\textit{Peggy}}]{};

      \node[individual] (h1) [xshift=0.5cm,above=1.7cm of userA,
      label=above:{\textit{Swimming}}]{};
      \node[individual] (h2) [xshift=0.5cm,above=1.7cm of
      userB,label=above:{\textit{Hiking}}]{};
      \node[individual] (h3) [xshift=0.5cm,above=1.7cm of
      userC,label=above:{\textit{Biking}}]{};
      \node[individual] (h4) [xshift=0.5cm,above=1.7cm of
      userD,label=above:{\textit{Rafting}}]{};
      \node[individual] (h5) [xshift=0.5cm,above=1.7cm of
      userE,label=above:{\textit{Jogging}}]{};
    \end{scope}

    \begin{scope}[very thick]
      \draw (userA) edge[shorten >=0.3cm,shorten <=0.3cm] (h1);
      \draw (userB) edge[shorten >=0.3cm,shorten <=0.3cm] (h2);
      \draw (userC) edge[shorten >=0.3cm,shorten <=0.3cm] (h3);
      \draw (userC) edge[shorten >=0.3cm,shorten <=0.3cm] (h4);
      \draw (userD) edge[shorten >=0.3cm,shorten <=0.3cm] (h4);
      \draw (userD) edge[shorten >=0.3cm,shorten <=0.3cm] (h5);
      \draw (userE) edge[shorten >=0.3cm,shorten <=0.3cm] (h5);
      \draw (userE) edge[shorten >=0.3cm,shorten <=0.3cm] (h3);
    \end{scope}
  \end{tikzpicture}
  \caption{Example of a social network graph exhibting various clones itmes.
    Edges connect a person with his or her activity. Equivalence
    classes for attribute clones are
    $\{\text{Swimming,Hiking}\},\{\text{Biking, Rafting, Jogging}\}$.}
  \label{fig:toyexample}
\end{figure}

\subsubsection{Data set description}
\label{sec:data-set-description}
Almost all of the following data sets can be obtained from the UCI Machine
Learning Repository~\cite{UCI}. We consider nine social network graphs and two
non social network data sets:
\begin{inparadesc}
\item[zoo]\cite{UCI}: 101 animals and seventeen attributes (fifteen Boolean and
  two numerical). All attributes were nominal scaled, resulting in a set with 43 attributes;
\item[cancer]\cite{wdbc}: 699 instances of breast cancer diagnoses with ten numerical attributes, which were nominal scaled;
\item[facebooklike]\cite{opsahl2009clustering}: 337 forum users with 522 topics they communicated on;
\item[southern]\cite{wasserman1994}: classical small world social network consisting of fourteen
  woman attending eighteen different events;
\item[club]\cite{brunson_club-membership}: 25 corporate executive officers and fifteen
  social clubs in which they are involved in;
\item[movies]\cite{faulkner2003music}: 39 composers of film music and their relations to 62
  producers;
\item[aplnm]\cite{Borchmann2017}: 79 participants of the \emph{Lange Nacht
    der Musik} in 2013 and the 188 events they participated in;
\item[jazz]\cite{gleiser03}: 198 jazz musicians and their collaborations;
\item[dolphin]\cite{Lusseau2003}: 62 bottlenose dolphins with contacts amongst each other;
\item[hightech]\cite{nr}: Some (one-mode) social network with 33 users
  from within the parameters of a social network but with no further
  insights provided;
\item[wiki]\cite{nr,leskovec2010signed}: 764 voters on Wikipedia with 605 users to be voted on.
\end{inparadesc}

For comparison, we also investigate randomized versions of all those data sets, generated using a coin draw process. This may imply that the resulting formal
contexts are prone to the stegosaurus phenomenon. However, no unbiased method
for generating formal context for a given number of objects, attributes, and
density is known~\cite{Borch16}.

\subsubsection{Computation}
Computing the attribute (object) clones for a given formal context $\GMI$ would
imply to know the associated attribute (object) closure system. However,
computing those is computational infeasible for contexts of a particular size or
greater. To cope with this barrier we
utilize~\cref{lem:clonirred} and~\cref{th:irreducible-intents}. Hence, instead of checking
all elements of a closure system we only need to check the
irreducibles. Therefore we checked brute force all combinations of attributes
(objects) for every given data set by checking the
according irreducibles.

In particular we computed for every data set the number of trivial and non-trivial
object clones, and attribute clones. The results are shown
in~\cref{tab:nets}. In addition we also computed the number of trivial and
non-trivial clones for the object/attribute-projections for every formal
context. However, besides creating more trivial clones no further
insights could be grasped from this. Also, the experiment on randomly
generated formal contexts had not different outcome. Therefore we
omitted presenting the particular results for the latter two.

\begin{table}[t]
  \centering
  {\small
    \caption{Properties of the considered (social) networks and data sets and
      results for clone experiment. With $G$-t we denote trivial clones
      whereas clones denote non-trivial clones. }
  \begin{tabular}{l|ccccccc}
    Name&$|U|$&$|M|$&density&\# $G$-clones&\# $M$-clones&\#  $G$-t-clones&\# $M$-t-Clones\\\midrule
    zoo&101& 43&0.390&0&0&42&2\\
    cancer&699&92&0.110&0&0&236&0\\
    facebooklike\, &377&522&0.014&7&0&24&83\\
    southern&18&14&0.352&0&0&1&1\\
    aplnm&79&188&0.061&0&0&1&21\\
    club&25&15&0.250&0&0&0&0\\
    movies&62&39&0.079&0&0&1&0\\
    jazz&198&198&0.068&7&7&0&0\\
    dolphins&62&62&0.082&0&0&2&2\\
    hightech&33&33&0.148&0&0&1&1\\
    wiki&764&605&0.006&234&234&73&30\\
  \end{tabular}}
  \label{tab:nets}
\end{table}

\subsubsection{Discussion}
The most obvious result for all data sets alike is that non-trivial clones are
very infrequent. Omitting the wiki data set only two data sets have clones at
all, in particular a very small number of object clones compared to the size of
the network. We investigated the exception by the wiki data set further and
discovered a large nominal scale as subcontext responsible for the vast amount
of clones. Since the wiki data set is the result of a collection of voting
processes this would represent single votes. For trivial clones we have diverse
observations. Some networks like facebooklike have a significant amount of
trivial clones. Others of comparable size, however, do not, like jazz. Since those
clones do not reveal any hidden structure but the fact that copies of users or
properties are present in the network, we consider these clones uninteresting.

For the object and attribute projections we obtain almost the same
results. Almost no non-trivial clones are present. Though, the number of trivial
clones has increased in almost all the networks. This could be another
indication that simple one-mode projections are insufficient for analyzing
bipartite networks.

All in all, the notion of non-trivial clones seems insufficient for
the investigation of graphs. The explanation for this is that the
structural requirements for two attributes being clone are too strong,
cf. theoretical results in~\cref{sec:theor-observ}.  However, it
strikes the question if there is a generalization which is softening
those requirements while preserving enough structure.


\section{Generalized Clones}
\label{sec:higher-order-clones}
The results from the previous section motivate finding a more general clone
notion for formal contexts. In~\cite{Gely05} the authors provided an interesting
generalization of clones in a formal context. They proposed $P$-Clones, \ie,
clones with respect to the family of pseudo intents, and $A$-Clones, \ie, clones
in a particular kind of atomized context. Both approaches are based on using
some kind of modified family of sets. Another course of action was taken
in~\cite{Macko12}, in which the author used a measure of \textquote{cloneniness}
based on the number of incorrect mapped sets.
We take a different approach, using the original set of closures -- the intents -- based on the following observation.
\begin{remark}[Clone permutation]
  Every pair $(a,b)$ of elements $a,b\in M$ with $a\sim b$ for a given formal
  context $\GMI$ gives rise to a permutation
  $\sigma:M\to M, m\mapsto \sigma(m)$, with $\sigma(a) = b,\ \sigma(b)=a,$ and $\sigma(m)=m$ for $ m\in M\setminus\{a,b\}$. We denote such permutations
  as \emph{clone permutations}.
\end{remark}

Since for every $a\in M$ we have $a\sim a$, the set of clone
permutations $S$ for a given formal context $\GMI$ contains the
identity. For any two elements $a,b\in M$ with $a\sim b$ we can
represent the associated clone permutation $\sigma$ by
$\sigma\coloneqq(ab)$ using the reduced cycle notation. From this we
note that the set of all pairs of proper clones corresponds to a
particular subset of permutations on $M$ where every permutation
$\sigma$ contains exactly one two-cycle. This gives rise to two
possible generalizations. Both associated computational problems
require sophisticated algorithms to be developed.

\subsubsection*{Multiple two-cycles}
\label{sec:longer-clones}
We motivate this approach using the lattice for a closure system on
$M=\{a,b,c,d\}$ represented in~\cref{fig:cf3} (left). In this closure
system there are no proper clones. However, we can find a permutation
$\sigma$ that preserves the closure system. For example, the
permutation $\sigma=(ab)(cd)$, which is a permutation of two disjoint
cycles of length two. This permutation is not representable by exactly
one cycle of length two. Hence, we propose permutations representable
as products of cycles of length two as one generalization of
clones. Yet, this immediately gives rise to the idea of higher order
permutations.



\subsubsection*{Higher Order}
\label{sec:higher-order}
Again, we want to motivate this generalization by providing an
example. In~\cref{fig:cf3} (middle), we show the lattice for a closure system
$M=\{a,b,c,d\}$. This closure system is free of (proper) clones. However, we
find a permutation $\sigma=(ab)(cd)$ in the above described manner. In addition
we find a permutation of order four, \ie, $\sigma^{4}=\mathop{id}$, preserving
the closure system, e.\,g., $\sigma=(acbd)$. In the same figure on the right we
observe a permutation of order five, i.e., $\sigma=(acedb)$, answering the
natural question for a permutation with odd order.

\begin{figure}[tb]
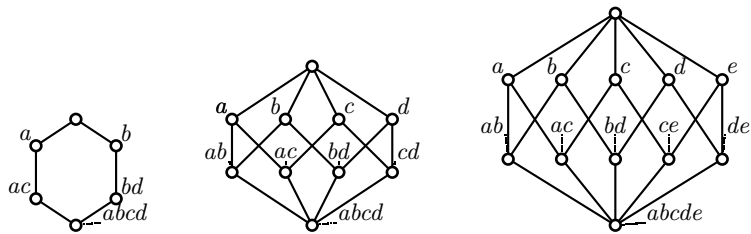

  \centering
  \input{lattice_examples.tex}\hspace{1cm}
  \input{lattice_crown4.tex}\hspace{1cm}
  \input{lattice_crown5.tex}
  \caption{Example for clone-free closure system on four attributes
    (left, middle) and on five attributes (right).}
  \label{fig:cf3}
\end{figure}

\section{Conclusion}
\label{sec:conclusion}
While starting the investigation the authors of this work were
confident to discover clones in graph data sets, at least for graphs of
a particular minimal size. In order to cope with the computational
complexity of closure systems we utilized results from~\cite{Gely05}
and expressed them in terms of statements about formal
contexts. However, our investigation did reveal the absence of clones
in real world graph like data. The only significant observation was
the emergence of trivial clones while projecting bipartite social
networks to one set of nodes.

This setback, though, led us to discover two more general notions of
clones, which can cope with more structural
requirements. Investigating those more thoroughly should be the next
step in clone related research, building on the theoretical results
in~\cref{sec:theor-observ}. To this end, we finish our work with
the following three open questions. \textbf{Question~1:} To which graph
theoretical notion could the idea of clone permutation correspond
to? \textbf{Question~2:} Does the set of all valid clone permutations
on a closure set always form a group and if no, why
not? \textbf{Question~3:} If yes, can this group provide new insights
into the structure of closure systems or of social networks?

\sloppy
{\small
\subsubsection*{Acknowledgments}
\label{sec:acknowledgments}

This work was funded by the German Federal Ministry
of Education and Research (BMBF) in its program
``Forschung zu den Karrierebedingungen
und Karriereentwicklungen des Wissenschaftlichen Nachwuchses
(FoWiN)'' under Grant 16FWN016.}

\medskip

\sloppy

\printbibliography

\end{document}